\documentclass[12pt, reqno]{amsart}

\usepackage{amsthm,amsmath,amssymb}
\DeclareMathAlphabet{\mathpzc}{OT1}{pzc}{m}{it}
\usepackage[margin=1.3in]{geometry}
\usepackage{setspace}
\usepackage[colorlinks=true, citecolor=black, linkcolor=black, urlcolor=black, setpagesize=false, plainpages=false, pdfpagelabels]{hyperref}

\newcommand{\conv}{\operatorname{conv}}
\newcommand{\den}{\operatorname{den}}

\newcommand{\lcm}{\operatorname{lcm}}
\newcommand{\Mm}{\mathpzc{M}}
\newcommand{\Po}{\mathcal{P}}
\newcommand{\psinv}{\dagger} 
\newcommand{\rank}{\operatorname{rk}}
\newcommand{\real}{\mathbb{R}}
\newcommand{\Sk}{\mathbb{S}^k}
\newcommand{\Ski}{\mathbb{S}^{k_i}}
\newcommand{\Skp}{\Sk_+}
\newcommand{\Sm}{\mathbb{S}^m}
\newcommand{\Smp}{\Sm_+}

\newcommand{\T}{\mathrm{T}}
\newcommand{\tr}{\operatorname{tr}}

\newcommand{\vA}{{A}}
\newcommand{\vb}{{b}}
\newcommand{\vB}{{B}}
\newcommand{\vc}{{c}}
\newcommand{\vC}{{C}}
\newcommand{\vD}{{D}}
\newcommand{\ve}{{e}}
\newcommand{\vf}{{f}}
\newcommand{\vI}{{I}}
\newcommand{\vk}{{k}}
\newcommand{\vK}{{K}}
\newcommand{\vM}{{M}}

\newcommand{\vP}{{P}}
\newcommand{\vQ}{{Q}}
\newcommand{\vR}{{R}}
\newcommand{\vs}{{s}}
\newcommand{\vu}{{u}}
\newcommand{\vV}{{V}}
\newcommand{\vW}{{W}}
\newcommand{\vx}{{x}}
\newcommand{\vX}{{X}}
\newcommand{\vY}{{Y}}
\newcommand{\vz}{{z}}
\newcommand{\vZ}{{Z}}

\allowdisplaybreaks[1]

\newtheoremstyle{myexamples}
{3pt}
{3pt}
{}
{}
{\bfseries}
{.}
{.5em}
{{\thmname{#1}\thmnumber{ #2}\thmnote{ (#3)}}}

\theoremstyle{plain}
\newtheorem{theorem}{Theorem}
\newtheorem{lemma}[theorem]{Lemma}
\newtheorem{proposition}[theorem]{Proposition}
\newtheorem{corollary}[theorem]{Corollary}
\theoremstyle{definition}
\newtheorem{definition}[theorem]{Definition}

\theoremstyle{myexamples}
\newtheorem{example}[theorem]{Example}

\title[Optimal designs for rational models]{Optimal designs for rational function regression}

\author{D{\'a}vid Papp}
\address{D{\'a}vid Papp, Northwestern University, Department of Industrial Engineering and Management Sciences, Technological Institute, 2145 Sheridan Rd, C210, Evanston, IL 60208. Email: \href{mailto:dpapp@iems.northwestern.edu}{dpapp@iems.northwestern.edu}}

\keywords{Optimal design, Approximate design, Rational function regression, Semidefinite programming, Linear matrix inequality}

\ifpdf\usepackage[final,expansion=true,protrusion=true]{microtype}\fi 

\begin{document}

\maketitle

\begin{abstract} We consider optimal non-sequential designs for a large class of (linear and nonlinear) regression models involving polynomials and rational functions with heteroscedastic noise also given by a polynomial or rational weight function. The proposed method treats D-, E-, A-, and $\Phi_p$-optimal designs in a unified manner, and generates a polynomial whose zeros are the support points of the optimal approximate design, generalizing a number of previously known results of the same flavor. The method is based on a mathematical optimization model that can incorporate various criteria of optimality and can be solved efficiently by well established numerical optimization methods. In contrast to previous optimization-based methods proposed for similar design problems, it also has theoretical guarantee of its algorithmic efficiency; in fact, the running times of all numerical examples considered in the paper are negligible. The stability of the method is demonstrated in an example involving high degree polynomials.
After discussing linear models, applications for finding locally optimal designs for nonlinear regression models involving rational functions are presented, then extensions to robust regression designs, and trigonometric regression are shown.
As a corollary, an upper bound on the size of the support set of the minimally-supported optimal designs is also found.
The method is of considerable practical importance, with the potential for instance to impact design software development. Further study of the optimality conditions of the main optimization model might also yield new theoretical insights.
\end{abstract}

\section{Introduction}

This paper is concerned with optimal approximate designs for polynomial and rational regression models with heteroscedastic error modeled by a rational weight function. In our focus is the general linear model
\begin{equation}\label{eq:model} y(t) = \sum_{i=1}^m \theta_i f_i(t) + \varepsilon(t), \qquad t \in \mathcal{I}, \end{equation}
where each $f_i$ is a known rational function defined on $\mathcal{I}$, and the error $\varepsilon(t)$ is a normally distributed random variable with mean zero and variance $\sigma^2(t) = 1/\omega(t)$, where the known \emph{weight function} $\omega$ is a rational function whose numerator and denominator are both positive on $\mathcal{I}$. We are interested in experiments designed to help estimate the unknown parameters $\theta_i$. The design space $\mathcal{I}$ is the finite union of closed, bounded intervals in $\real$, also allowing singletons as degenerate intervals. We assume that observations are uncorrelated, and that the $f_i$ are linearly independent.

Our main result is a characterization of the support of the D-, E-, A-, and $\Phi_p$-optimal designs as the optimal solutions of a semidefinite optimization problem. This directly translates to a method to numerically determine the optimal design, using readily available optimization software. The characterization is applicable to every linear model involving polynomials and rational functions with heteroscedastic noise also given by a polynomial or rational weight function. We demonstrate that the method is numerically robust (in the sense that it can handle ill-conditioned problems, such as those involving polynomials of high degree), and has very short running time on problems of practical size.

Optimal designs for Fourier regression models and locally optimal designs for certain nonlinear models can also be found with similar methods.

In many cases the experimenter is interested only in certain linear combinations of the parameter vector $\theta := (\theta_1, \dots, \theta_m)^\T$, which are given by the components of $\vK^\T\theta$ for some $m \times s$ matrix $\vK$. In the presentation of our approach it is convenient to assume that our goal is to estimate the entire parameter vector, that is $\vK=\vI_m$ (the $m \times m$ identity matrix), and that the design space contains enough points to make all parameters estimable. (If $\vK=\vI_m$, the latter assumption means that there is a design whose information matrix is non-singular, see later.)  In \mbox{Section \ref{sec:subsystems}} we show how the proposed method can be generalized to handle problems with general $\vK$.

Much attention has been devoted to optimal designs for special cases of model \eqref{eq:model}. It is well known that when the design space $\mathcal{I}$ is finite, the D-, E-, and A-optimal approximate designs can be found by convex optimization even for arbitrary $f_i$'s, see, for example \cite[Chapter 7]{BV-04}, or a generalization of this approach to multi-response experiments in \cite{AS-09}. However, when $\mathcal{I}$ is an interval, considerable difficulties arise, as the finite support of the optimal design also has to be characterized.

A popular approach in the literature is that a polynomial is sought whose roots are the support points of the optimal design. For instance, as discovered by Guest \cite{Guest-58} and Hoel \cite{Hoel-58}, the D-optimal design for ordinary polynomial regression, when $f_i = t^i$, and $\omega$ is a positive constant, on $\mathcal{I}=[-1,1]$ is the one that assigns uniform weights to each of the zeros of $t\to (1-t^2)\tfrac{d}{dt}L_m(t)$, where $L_m$ is the Legendre polynomial of degree $m$. The number of support points had already been determined in \cite{dlG-54}. Similar characterizations are known for A- and E-optimal designs for polynomial regression, see, for example the classic monographs \cite{Fedorov-72, Pukelsheim-93}. Another common approach is to determine the canonical moments of the optimal design \cite{DR-96, DS-97}. Further optimality criteria for polynomial models, and closed-form characterizations of the optimal designs for linear and quadratic models, are discussed in \cite{Stigler-71}. See also \cite{IS-01} for E-optimal designs for linear models with rational functions $f_i(t) = (t-\alpha_i)^{-1}$ with $\alpha_i \not\in \mathcal{I}$. The Optimum Experimental Design website \cite{design-website} also contains a rather comprehensive list of solved models, along with an impressive, and continuously maintained, list of references.

More recently considerable attention has been paid to polynomial models with missing interactions, also called \emph{incomplete} or \emph{improper} polynomial models. Representative results include \cite{Dette-92}, which gives D-optimal designs when only odd or only even degree terms appear in the model; \cite{HCW-95} and \cite{CH-96}, which consider D- and E-optimal designs (respectively) for polynomial models with zero constant term; \cite{DR-96}, which considers D-optimal designs, also for some multivariate problems, over the unit cube under less restrictive assumptions on the missing terms; and \cite{Fang-02}, which gives D-optimal designs when only the lowest degree terms, up to a fixed degree $m^\prime$, are absent. Note that even the union of these methods does not yield a complete solution to incomplete polynomial models, even for univariate regression with homoscedastic error.

Results in the heteroscedastic case are even more scarce and typically less general. For instance, \cite{IKS-98} is devoted to D-optimal designs for polynomial regression over $[0,b]$ with the weight function $\omega(t)= t/(1+t)$.

The design space $\mathcal{I}$ is almost always a (closed, bounded) interval, which is probably sufficient for most applications. Imhof and Studden \cite{IS-01} also considered some rational models when $\mathcal{I}$ is the union of two disjoint intervals.

Most of the above results are based on the theory of orthogonal polynomials, canonical moments \cite{DS-97}, and Chebyshev systems \cite{KS-66}. They are rather specific in their scope, and generalization of their proofs appears to be difficult. On the other hand, most of them yield numerically very efficient methods for computing numerically optimal designs. The bottleneck in these methods is either polynomial root-finding, which can be carried out in nearly linear time in the degree of the polynomial \cite{Pan-01}, or the reconstruction of a measure on finite support from its canonical moments, which can also be carried out relatively easily \cite{DS-97}. An exception is the method of \cite{Fang-02}, which involves finding the global maximum of a \emph{multivariate} polynomial (even though it is concerned with univariate polynomial regression only). This is an NP-hard problem even in very restricted classes of polynomials, and is known to be very difficult to solve in practice even when the number of variables and the degree are rather small \cite{gloptipoly}.

In the pursuit of more widely applicable methods, some of the attention has turned to the numerical solution of optimization models that characterize optimal designs. Pukelsheim's monograph \cite{Pukelsheim-93} is a comprehensive overview of optimal design problems with an optimization-oriented viewpoint, but it is not concerned with algorithms or numerical computations. Most numerical methods proposed in the literature are variants of the popular \emph{coordinate-exchange method} from \cite{MN-95}, which is a variant of the classic Gauss--Seidel method (also known as coordinate descent method) used in derivative-free optimization. These algorithms maintain a finite working set of support points, and iteratively replace one of the support points by another one from $\mathcal{I}$ if the optimal design on the new support set is better than that of the current support set. See \cite{CL-07} for a recent variant of this idea for finding approximate D-optimal designs.

However, this approach has serious drawbacks, and care has to be taken not to abuse them: (i) some variants require that the size of the minimally supported optimal design be known \emph{a priori}; (ii) no bound is known on the number of iterations the algorithm might take; (iii) in fact, the number of iterations of the coordinate descent method is known to be quite high in practice even for some very simple convex optimization problems \cite[Chapter 9]{NW-00}; and (iv) the coordinate descent method does not necessarily converge at all if the function being optimized is not continuously differentiable \cite{Rusz-05}. Hence, these methods can hardly be considered a completely satisfactory solution of most polynomial regression problems, even though some successful numerical experiments have been reported, cf. \cite{CL-07}.

This paper proposes a different approach to linear regression models involving polynomials and rational functions. Motivated in part by the approach of \cite{BV-04}, it is also based on an optimization model involving linear matrix inequalities, which can be solved efficiently, both in theory and in practice, by readily available optimization software.

The novelty of the proposed method is that it does not work with the support points directly, as existing numerical methods, such as the coordinate-exchange method, do. Instead, it follows some of the previous symbolic approaches by computing the coefficients of a polynomial whose zeros are the support points of the optimal design.

After introducing the problem formally, we derive our main theorems in Section \ref{sec:design}
for the estimation of the full parameter vector $\theta$. Illustrative examples are presented in \mbox{Section \ref{sec:examples}}. \mbox{Section \ref{sec:subsystems}} is concerned with the more general case, when only a subset of the parameters (or their linear combinations) need to be estimated.
We then apply these results to finding locally optimal designs for nonlinear models in \mbox{Section \ref{sec:nonlinearmodels}}. Finally, in \mbox{Section \ref{sec:othersystems}} we give an outlook to models of regression involving other functions than rational functions.

\paragraph{Notation} We will make use of the following, mostly standard, notations: $\deg p$ denotes the degree of the polynomial $p$, $\lcm$ stands for the least common multiple of polynomials. The denominator of a rational function $r$ is denoted by $\den(r)$. The positive part function is denoted by $(\cdot)_+$. The brackets $\langle \cdot, \cdot \rangle$ denote the usual (Frobenius) inner product of vectors and matrices, that is, $\langle \vA, \vB \rangle = \sum_{ij}A_{ij}B_{ij}$. Since many decision variables in the paper are matrices, linear constraints on matrices are written in operator form. For example, a linear equality constraint on an unknown matrix $\vX$ will be written as $A(\vX)=\vb$ (where $A$ is a linear operator and $\vb$ is a vector) to avoid the cumbersome ``vec'' notation necessary to use matrix-vector products. For the linear operator $A$,  $A^*$ denotes its adjoint. The identity operator is written as $\operatorname{id}$.

The space of $m\times m$ symmetric matrices is denoted by $\Sm$, the cone of $m\times m$ positive semidefinite real symmetric matrices is $\Smp$. The \emph{L{\"o}wner partial order} on $\Sm$, denoted by $\succcurlyeq$, is the conic order generated by $\Smp$; in other words, we write $\vA \succcurlyeq \vB$ when $\vA-\vB \in \Smp$.

\section{Optimality criteria and their semidefinite representations}

A \emph{design for infinite sample size} (also called \emph{approximate design} or \emph{design} for short) is a finitely supported probability measure $\xi$ on $\mathcal{I}$. Using the notation $\vf(t) = (f_1(t), \dots, f_m(t))^\T$, the \emph{Fisher information matrix} of $\theta$ corresponding to the design $\xi$ is
\begin{equation}\label{eq:Fisher} \vM(\xi) = \int_{\mathcal{I}} \vf(t) \vf(t)^\T \omega(t) d \xi(t).\end{equation}
Of course, this integral simplifies to a finite sum for every design. Note that for every $\xi$, $\vM(\xi) \in \Smp$. A design $\hat\xi$ is considered \emph{optimal} if $\vM(\hat\xi)$ is maximal with respect to the L{\"o}wner partial order (recall the end of the previous section); see \cite[Chapter 4]{Pukelsheim-93} for detailed statistical interpretation. If $\Phi$ is an $\Smp\to\real$ function, the design $\hat\xi$ is called \emph{optimal with respect to $\Phi$}, or \emph{$\Phi$-optimal} for short, if $\Phi(\vM(\hat\xi))$ is maximum. Again, only those criteria are interesting which are compatible with the L{\"o}wner partial order, that is functions $\Phi$ satisfying $\Phi(\vA) \geq \Phi(\vB)$ whenever $\vA \succcurlyeq \vB \succcurlyeq 0$. Popular choices of $\Phi$ include the following.
\begin{enumerate}
    \item When $\Phi(\vM) = \det(\vM)$, $\hat\xi$ is called \emph{D-optimal}.
    \item When $\Phi(\vM) = \lambda_1(\vM)$, the smallest eigenvalue of $\vM$, $\hat\xi$ is called \emph{E-optimal}.
    \item When $\Phi(\vM) = -\tr(\vM^{-1})$, where $\tr$ denotes matrix trace, $\hat\xi$ is called \emph{A-optimal}.
    \item When $\Phi(\vM) = (\tr(\vM^p))^{1/p}$, $\hat\xi$ is called \emph{$\Phi_p$-optimal}.
\end{enumerate}
For most purposes of the paper $\Phi$ could be an arbitrary concave extended real valued function on $\Smp$ with finite values on the interior of $\Smp$. However, to avoid certain technical difficulties, and in order to obtain good characterizations of optimal designs, we will assume that the $\Phi$ of our choice is \emph{representable by linear matrix inequalities} (LMIs) or \emph{semidefinite representable}, this includes all of the criteria discussed above. The precise definitions we need are summarized next.

\begin{definition}\label{def:SD-representable-set}
A set $S\subseteq\real^n$ is \emph{semidefinite representable} if for some $k\geq 1$ and $l \geq 0$ there exist affine functions $A:\real^n\to\Sk$ and $C:\real^l\to\Sk$ such that the set $S$ can be characterized by a linear matrix inequality in the following way:
\[ S = \{\vs \in \real^n\;|\; \exists\, \vu \in \real^l \colon A(\vs) + C(\vu) \succcurlyeq 0 \}. \]
\end{definition}
Note that the intersection of semidefinite representable sets is also semidefinite representable, so we could equivalently allow to have a characterization of the above form with $p \geq 1$ inequalities. The motivation behind the idea of semidefinite representable sets is that finding global optima of ``nice'' functions over them is easy, and a number of numerical methods are available to that in an efficient manner. ``Nice'' functions include semidefinite representable functions, defined below, in \mbox{Definition \ref{def:SD-representable-func}}.

In this paper we will encounter two important instances of semidefinite representable sets: the coefficient vectors of polynomials that are nonnegative over an interval, and the level sets of the optimality criteria $\Phi$.

\begin{lemma}[\protect{\cite[Chapter 2]{KS-66}}]\label{lem:pospoly-sd-rep}
The set
\[P_d^{[a,b]} = \Big\{ (p_0, \dots, p_d) \colon \sum_{i=0}^d p_i x^i \geq 0\;\;\forall\,x\in[a,b] \Big\} \]
of coefficient vectors of polynomials of degree $d$ that are are nonnegative over the interval $[a,b]$ is a semidefinite representable subset of $\real^{d+1}$.
\end{lemma}
This is a reasonably well known theorem in probability and statistics owing to its application in moment problems \cite{DS-97}, for completeness we provide a specific representation in the Appendix. The same assertion holds even if the polynomials are represented in another basis, not in the monomial basis, but the actual characterization will, of course, be different.

The next definition is necessary to define the class of optimality criteria our approach can handle.
\begin{definition}\label{def:SD-representable-func} A function $\Phi:\Smp\to\real$ is \emph{semidefinite representable} if its (closed) upper level sets are semidefinite representable, that is, if for some $k_1, \dots, k_p$ and $l$ there exist linear functions $A_i:\Smp\to\Ski$, $C_i:\real^l\to\Ski$, and matrices $\vB_i \in \Ski$, $\vD_i\in\Ski$ $(i=1, \dots, p)$ such that for all $\vX\in\Smp \text{ and } z\in\real$,
$\Phi(\vX) \geq z$ holds if and only if \begin{equation}\label{eq:SDP-rep_def}A_i(\vX) + \vB_iz + C_i(\vu) + \vD_i \succcurlyeq 0 \quad i = 1, \dots, p\end{equation}
for some $\vu \in \real^l$.
\end{definition}

As mentioned above, finding the optimal value (and the optimizer) of a semidefinite representable function over a semidefinite representable set is generally easy; optimization problems of this form are called \emph{semidefinite optimization problems} or \emph{semidefinite programs}; see also the beginning of the next section.

We will also need the following (technical) assumption on the relationship between the model (as defined by the functions $f_i$ and $\omega$) and the criterion function $\Phi$. It is only used in the proof of the main theorem.

\begin{definition}
We say that the semidefinite representable function $\Phi:\Smp\to\real$ is \emph{admissible} with respect to the set $\mathcal{X} \subseteq \Smp$ if $\Phi$ has a representation \eqref{eq:SDP-rep_def} for which there exists an $\hat{\vX} \in \mathcal{X}$ satisfying \eqref{eq:SDP-rep_def} with strict inequality for some $z$ and $\vu$. That is to say that the left-hand side of each of the $p$ inequalities can be made positive definite simultaneously for at least one $\hat{\vX} \in \mathcal{X}$.
\end{definition}

This is a rather technical condition in the sense that most interesting functions $\Phi$ are admissible with respect to every non-empty set $\mathcal{X}$ (a sufficient condition for this is that in the semidefinite representation of $\Phi$ each $\vB_i$ be positive or negative definite), or at least with respect to every $\mathcal{X}$ that contains a non-singular matrix.

D-, E-, and A-optimality are all semidefinite representable, or are equivalent to other criteria given by semidefinite representable functions. The same holds for $\Phi_p$-optimality. They are also admissible with respect to every set of Fisher information matrices for which the criteria is well-defined (see below). Note that all semidefinite representable functions are quasi-concave, continuous functions.

\begin{example}[E-optimality]\label{ex:E-optimality}
For every $\vM \in \Sm$, $\lambda_1(\vM) \geq z$ if and only if $\vM - z\vI \succcurlyeq 0$, so $\lambda_1$ admits a simple semidefinite representation. In this representation $p=1$, $A_1 = \operatorname{id}$, $\vB_1 = -\vI$, $C_1 \equiv 0$, and $\vD_1 = 0$, hence $\lambda_1$ is admissible with respect to every non-empty set of Fisher information matrices.
\end{example}

\begin{example}[A-optimality]\label{ex:A-optimality}
It follows from Haynsworth's theorem \cite{Haynsworth-68} on the inertia of Hermitian block matrices that a symmetric block matrix $\left(\begin{smallmatrix}\vP&\vQ\\ \vQ^\T&\vR\end{smallmatrix}\right)$ with positive definite block $\vP$ is positive semidefinite if and only if its \emph{Schur complement}, given by $\vR - \vQ^\T \vP^{-1}\vQ$, is positive semidefinite. Let $\vM \in \Smp$ be invertible, for example an invertible Fisher-information matrix, and fix a $k \in \{1, \dots, m\}$. Plugging in $\vM$ for $\vP$, the $k$th unit vector $\ve_k$ for $\vQ^\T$, and a scalar $u$ for $\vR$ we have that $(\vM^{-1})_{k,k} \leq u$ if and only if $\left(\begin{smallmatrix}\vM&\ve_k\\ \ve_k^\T& u\end{smallmatrix}\right) \succcurlyeq 0$. This observation yields a semidefinite representation of A-optimality of the form \eqref{eq:SDP-rep_def} with $p=m+1$:
\[\tr(\vM^{-1}) \leq z \quad \text{iff} \quad \exists\, u_1,\dots,u_m\colon z \geq \sum_i u_i,\; \text{and}\; \left(\begin{smallmatrix}\vM&\ve_k\\ \ve_k^\T& u_k\end{smallmatrix}\right) \succcurlyeq 0, k=1,\dots, m.\]

It follows that the A-optimality criterion is admissible with respect to every set of Fisher information matrices that contains at least one non-singular matrix.
\end{example}

%

\begin{example}[D- and $\Phi_p$-optimality]
The cases of D-optimality and $\Phi_p$-optimality are more complicated, but can also be fitted in the above framework. Owing to page limitations we can only give the flavor of this result, and pointers to the literature.

D-optimality is equivalent to optimality with respect to the criterion $\Phi(\vM) = (\det(\vM))^{1/m}$, where $m$ is the size of $\vM$. Note that this is the geometric mean of the eigenvalues of $\vM$. $\Phi_p$-optimality is expressed by the matrix mean $\Phi_p(\vM) = \bigl(\tr(\vM^p)\bigr)^{1/p} = \bigl(\sum_{i=1}^m\lambda_i^p\bigr)^{1/p}$, where $\lambda_i$ is the $i$th eigenvalue of $\vM$. Hence, both criteria are symmetric functions of the eigenvalues of $\vM$. Moreover, both the geometric mean and the $p$-norm, for every rational $p\geq 1$ are also semidefinite representable \cite[Section 3.3.1]{BTN-01}. Finally, we can invoke \cite[Proposition 4.2.1]{BTN-01}, which states that for every semidefinite representable symmetric $g\colon\real^m\to\real$, the function $\Phi(\vM) = g\bigl(\lambda_1(\vM), \dots, \lambda_m(\vM)\bigl)$ is also semidefinite representable.

D- and $\Phi_p$-optimality are also admissible with respect to every set of Fisher information matrices that contains at least one non-singular matrix.
\end{example}

Another interesting optimality criterion, not considered in this paper, is the maximin efficient criterion. Models for which maximin efficient approximately optimal designs can be found using semidefinite programming (this includes polynomial models) can be found in the recent technical report \cite{FT-10}.

\section{Optimal designs and semidefinite optimization}\label{sec:design}

First we shall give a very short introduction to semidefinite optimization to summarize the background necessary to keep this paper self-contained. The reader is also encouraged to consult \cite{VB-96}; or \cite{WSL-00} for a considerably more in-depth survey to this vast field.

Semidefinite optimization (or semidefinite programming) is a generalization of the familiar linear optimization. A \emph{semidefinite program} (or SDP for short) is the mathematical problem of finding the optimum of a linear function subject to the constraint that an affine combination of matrices is positive semidefinite. In other words, it is an optimization problem of the form
\begin{equation}\label{eq:SDP-def}\begin{aligned}
\mathop{\text{minimize}}_{\vx \in \real^n}\;\; & \sum_i c_i x_i \\
\text{subject to}\;\; & \vA_0 + \sum_{i=1}^n \vA_i x_i \succcurlyeq 0,
\end{aligned}
\end{equation}
where $\vc \in \real^n$ and $\vA_i \in \Sm$, $(i=0, \dots, n)$ are given; $x_i$ denotes the $i$th component of the vector $\vx$; these are the variables.

Constraints of the above form are called \emph{semidefinite constraints} or \emph{linear matrix inequalities}. The format of problem \eqref{eq:SDP-def} is regarded as a ``standard form'', but other, seemingly more general optimization problems that can be converted to the above form are also considered semidefinite programming problems. In particular, multiple semidefinite constraints can be added to the problem, and the constraints can be augmented by linear inequalities and equations, as these translate to constraints on diagonal matrices. Matrices of variables can also be considered, and constrained simultaneously in the form $L(\vX) \succcurlyeq \vC$; here $\vX$ is the matrix of variables, $L$ is a linear operator, and $\vC$ is a matrix of appropriate size. More generally, the maximization of every semidefinite representable function over every semidefinite representable set (as defined in the previous section) can be cast as an SDP. In this paper we will show that finding the support of the optimal design can be cast as an SDP of this more general form, for every regression model \eqref{eq:model}.

Semidefinite programs are special convex optimization problems, and the standard duality theory of convex optimization \cite{rockafellar-70,Rusz-05} applies to them. Algorithms to numerically compute the optimal solutions of a semidefinite program have been well studied for more than two decades; SDPs involving tens of thousands of variables are routinely solved in the literature \cite{VB-96}. The SDPs of this paper are considerably smaller; they can be solved in a fraction of a second without any numerical issues by commonly used SDP solver software, such as SeDuMi \cite{sedumi}, a freely available Matlab toolbox. Additional toolboxes, such as CVX \cite{cvx} and YALMIP \cite{yalmip}, are available to translate ``high-level'' semidefinite programs involving semidefinite functions such as the optimality criteria mentioned in this paper to the semidefinite programs in the above ``standard'' form.

\subsection{Semidefinite representation of optimal designs}

The main result in this section, and of the paper, is that the problem of finding an optimal design with respect to $\Phi$ can be equivalently written as a semidefinite programming problem whenever the functions $f_i, i=1, \dots, m$ and $\omega$ are rational functions defined over a finite union $\mathcal{I}$ of closed intervals, and $\Phi$ is a semidefinite representable function that satisfies the mild technical condition that it is admissible with respect to the set of all Fisher information matrices.

As mentioned in the Introduction, this is already known for design spaces $\mathcal{I}$ consisting of finitely many points, even for arbitrary $\{f_i\}$. While it is not stated there in this general form, the following theorem is implicit in \cite[Chapter 7]{BV-04}:

\begin{theorem}[\protect{\cite[Chapter 7]{BV-04}}] \label{thm:eperfagyi}
Let $\mathcal{I} \subset \real$ be finite, and $\Phi$ be a semidefinite representable function compatible with the L{\"o}wner partial order. Then the  $\Phi$-optimal designs for model \eqref{eq:model} are characterized as the set of optimal solutions to a semidefinite programming problem.
\end{theorem}

In this semidefinite programming problem the support points are fixed parameters, and the variables are the masses the optimal design assigns to the support points; hence \mbox{Theorem \ref{thm:eperfagyi}} allows us to find the optimal design only once its support is known. Treating the support points as variables would be problematic for two reasons: the number of support points for the optimal design may not be known, and even if it was, the resulting optimization problem would be intractable. Our goal in this paper is to characterize the support of the optimal design as a solution of a semidefinite program. In the optimization problem we are about to define, the variables are the \emph{coefficients of a polynomial whose roots are the support points of the optimal design}.


Our main result, \mbox{Theorem \ref{thm:main}} below, is the characterization of the support of the optimal design as a solution of a semidefinite program. After finding the support, \mbox{Theorem \ref{thm:eperfagyi}} can be applied to find the weights---by solving another semidefinite program.


\begin{theorem}\label{thm:main}
Suppose that in the linear model \eqref{eq:model} $\mathcal{I}$ is a finite union of closed intervals, the functions $f_i$ are rational functions with finite values on $\mathcal{I}$, and $\omega$ is a nonnegative rational function on $\mathcal{I}$. Let $\Phi$ be an admissible semidefinite representable function (with representation \eqref{eq:SDP-rep_def}) with respect to the set of Fisher information matrices $\Mm = \conv\{ \vf(t)\vf(t)^\T\omega(t)\,|\,t\in \mathcal{I} \}$. Then the support of the $\Phi$-optimal design is a subset of the real zeros of the polynomial $\pi$ obtained by solving the following semidefinite programming problem:
\begin{subequations}\label{eq:final}
\begin{align}
\mathop{\operatorname{minimize}}_{\substack{y\in\real, \pi\in\real^d,\\ \vW_1, \dots, \vW_p \in \Skp}}\;\; & y \\
\operatorname{subject\,to}\;\; & \sum_{i=1}^p \langle \vW_i, \vB_i \rangle = -1, \quad \sum_{i=1}^p C_i^*(\vW_i) = 0,\label{eq:final-BC}\\
                      & \pi = \Pi(y, \vW_1, \dots, \vW_p),\label{eq:PI-constr} \\
                      & \pi \in P_d^\mathcal{I},\label{eq:POP-constr}
\end{align}
\end{subequations}
where $d$ is the degree of the polynomial
\begin{equation}\label{eq:def_pi}
t \to \lcm(\den(\omega), \den(f_1^2), \dots, \den(f_p^2)) \bigg(y - \sum_{i=1}^p \langle \vW_i, A_i(\vM(\xi_t))+\vD_i\rangle\bigg),
\end{equation}
whose coefficient vector is denoted by $\Pi(y, \vW_1, \dots, \vW_p)$ in \eqref{eq:PI-constr} above.
\end{theorem}

Note that the operator $\Pi$ in \eqref{eq:def_pi} is affine, hence aside from \eqref{eq:POP-constr} every constraint in \eqref{eq:final} is a linear equation or linear matrix inequality. Furthermore, \eqref{eq:POP-constr} can be translated to linear matrix inequalities using \mbox{Lemma \ref{lem:pospoly-sd-rep}}. Hence, \eqref{eq:final} is indeed a semidefinite program.

Not wanting to defer the discussion of examples and extensions, the proof was moved to the Appendix. Instead, we discuss a few examples.

\section{Examples}\label{sec:examples}

We start with two detailed examples demonstrating how E- and A-optimal design problems translate to semidefinite optimization models. Then the numerical robustness of the proposed method is investigated using a high degree polynomial model. Finally, an example with rational models is shown, in which the parameters of point sources emitting radiation are estimated from measurements of total intensity.

All timing results in this paper were obtained using the semidefinite solver SeDuMi \cite{sedumi} running on an ordinary desktop computer with a 2.83GHz processor, using a single core.

\begin{example}[E-optimal designs without an intercept]\label{ex:numerical1}
This problem was considered in \cite{CH-96}, and we use it here to illustrate the steps of the approach and to verify the correctness of our model in a relatively high degree model that has been solved: $\mathcal{I}=[-1,1]$, $f_i = t^i$, $i=1, \dots, m$, and $\omega$ is a positive constant. Using the semidefinite representation of E-optimality given in \mbox{Example \ref{ex:E-optimality}}, the variables in the optimization model of \mbox{Theorem \ref{thm:main}} are the scalar $y$ and the positive semidefinite matrix $\vW_1$ of order $m$. The constraints can be derived as follows: from \mbox{Example \ref{ex:E-optimality}} we have $A_1 = \operatorname{id}$, $\vB_1 = -I$, $C_1 \equiv 0$, and $\vD_1 = 0$, hence $\langle \vB_1,\vW_1 \rangle = -\tr(\vW_1)$, and $\langle \vC_1,\vW_1 \rangle = 0$. Also note that 
\[\vM(\xi_t) = (t, t^2, \dots, t^m)^\T (t, t^2, \dots, t^m) = \left(\begin{smallmatrix}t^2 & t^3 & \cdots & t^{m+1}\\t^3 & t^4 & \cdots & t^{m+2}\\
\vdots & & & \vdots \\
t^{m+1} & t^{m+2}& \cdots & t^{2m}
\end{smallmatrix}\right),\]
where the last matrix has $t^{i+j}$ as its $(i,j)$-th entry.

Hence, the first constraint of \eqref{eq:final-BC} is $\tr(\vW_1) = 1$, whereas the second constraint of \eqref{eq:final-BC} is simply $0=0$, and can be omitted. We have $\deg(\pi)=2m$, and the correspondence between the entries of $\vW_1$ and the coefficients of $\pi(t) = \sum_{i=0}^{2m} p_i t^i$, given in \eqref{eq:def_pi}, simplifies to the system of equations
\[ p_0 = y, \; p_1 = 0,\;\text{and}\; p_{k} = -\sum_{i+j=k} (\vW_1)_{ij}\; \text{for}\;k=2,\dots,2m .\]
In summary, dropping the subscript from $\vW_1$, we have the optimization problem
\begin{align*}
\mathop{\text{minimize}}_{y\in\real, \pi\in\real^{2m}, \vW\in\Smp}\;\; & y \\
\text{subject to}\;\; & \tr(\vW) = 1, \\
                      & \pi = (y, 0, S_2, \dots, S_{2m}) \in P^{[-1,1]},
\end{align*}
where $S_{k} = -\sum_{i+j=k} W_{ij}$ $(k=2, \dots, 2m)$ are the anti-diagonal sums of the matrix $\vW$, and the constraint $\pi = (y, 0, S_2, \dots, S_{2m}) \in P^{[-1,1]}$ can be turned into the system of linear and semidefinite constraints \eqref{eq:EvenDegreePolyChar} given in the Appendix, plugging in $a=-1, b=1$.

For practical computations several Matlab toolboxes, such as CVX \cite{cvx} and YALMIP \cite{yalmip}, are available to facilitate the translation of semidefinite programs such as the one above to the the so-called ``standard form'' required by semidefinite solvers. Rather than providing a detailed description or comparison of these programs, we offer a completely self-explanatory example, the formulation of the above problem in the language of the CVX toolbox, in \mbox{Figure \ref{fig:CVX_1}}. Note that both the trace constraint and the nonnegative polynomial constraint are represented at the same high level in the code as in the mathematical model above. They are translated to a standard form semidefinite program and solved using a semidefinite programming solver automatically by CVX, leaving virtually no work to the user.

\begin{figure}
\small
\begin{verbatim}
cvx_begin
    m = 8;
    variable y;
    variable W(m,m) symmetric;
    variable pi(2*m+1);
    
    minimize y;
    subject to
        W == semidefinite(m);
        trace(W) == 1;
        pi(1) == y;
        pi(2) == 0;
        -pi(3)  == W(1,1);
        ...
        -pi(17) == W(8,8);

        pi(end:-1:1) == nonneg_poly_coeffs(2*m, [-1,1]);
cvx_end
\end{verbatim}
\normalsize
\caption{Matlab solution for Example \ref{ex:numerical1} using the CVX toolbox. Note that Matlab indexes vectors starting from 1 instead of 0. The equations defining the coefficients $\pi_4$ through $\pi_{16}$ have been omitted for brevity.} \label{fig:CVX_1}
\end{figure}

For example, solving the resulting problem for $m=8$, the optimal vector $\pi$ is the coefficient vector of a degree 16 polynomial whose real roots are: $\{\pm 1, \pm 0.9207, \pm 0.693, \pm 0.3357\}$. It also has two imaginary roots. The eight real roots constitute the support of the E-optimal design. The same numerical example was considered in \cite{CH-96} with, of course, the same conclusion. The running time of SeDuMi in this example was 0.2 seconds.
\end{example}

\begin{example}[A heteroscedastic polynomial model]
Consider the cubic model $f_i = t^{i-1}$, $i=1,\dots,4$, with heteroscedastic noise given by $\omega(t)=1/(1+t^2)$, over the design space $[-5,5]$. We chose this arbitrary model because it is one of the simplest among those whose solution appears to not to be characterized.

The A-optimal design is computed as follows. The parameters $A_i, \vB_i, C_i, \vD_i$ in the semidefinite representation \eqref{eq:SDP-rep_def} of A-optimality are determined first from \mbox{Example \ref{ex:A-optimality}}. Using this representation, the constraints of the semidefinite programming problem in \mbox{Theorem \ref{thm:main}} are compiled in the following way.
\begin{itemize}
\item There are 4 semidefinite matrices $\vW_1,\dots,\vW_4$ of order $5$, and $\vW_5$ is a nonnegative scalar.
\item The first constraint of \eqref{eq:final-BC} is simply $\vW_5=1$. The second constraint of \eqref{eq:final-BC} translates to $(\vW_i)_{5,5}=1$ for each $i=1,\dots,4$.
\item We have $\deg(\pi)=6$, and comparing the coefficients on the two sides of \eqref{eq:def_pi}, we obtain a linear system of equations and matrix inequalities for \eqref{eq:PI-constr} and \eqref{eq:POP-constr}, along the same lines as in the previous example.
\end{itemize}
The optimal solution is a polynomial whose real roots are $\{\pm5, \pm0.854\}$, this is the support of the A-optimal design.
\end{example}

In the remaining examples we shall refrain from the detailed list of the above steps, and concentrate on the main features of the models and the numerical results.

\begin{example}[Polynomial models of high degree]
We now consider the problem of designing experiments for very high degree polynomial models in order to test the numerical stability and scalability of our approach. Models involving high degree polynomials are rarely justifiable, but they are good problems to test numerical stability, as they are notoriously ill-conditioned. For example, in the basic polynomial model, when $f_i = t^{i-1}, i=1, \dots, n$, the the Fisher information matrix $\vM(\xi)$ in \eqref{eq:Fisher} becomes a Hankel matrix, which is known to be ill-conditioned \cite{Tyrtyshnikov94}. Also note that in the case of rational models, the polynomial defined by \eqref{eq:def_pi} might also have a degree that is considerably higher than the degree of the numerators and denominators of the functions $f_i$, leading to potentially ill-conditioned optimization models.
The numerical difficulties can be somewhat alleviated by using an orthogonal polynomial basis in \eqref{eq:model}. In this example we look for the E-optimal polynomial design in the ordinary polynomial regression model, but using the Legendre polynomial basis: $f_i=P_{i-1}$, the $(i-1)$-st Legendre polynomial defined by $P_0 = 1, P_1(t)=t$, and $P_{n+1}(t)=(n+1)P_{n+1}(t)+(2n+1)tP_{n}(t)-nP_{n-1}(t)$ for $n\geq 1$.

The constraints are obtained along the same lines as in \mbox{Example \ref{ex:numerical1}}, except that the coefficients of $\pi$ in \eqref{eq:def_pi} need to be changed as the moment matrix $\vM(\xi_t)$ changes with the change of basis.

We solved the resulting semidefinite program for the degree $20$ model; the computation required 0.4 seconds. The optimal polynomial $\pi$ is a nonnegative polynomial on $[-1,1]$ with single roots at the endpoints $\pm 1$, and double real roots at the points \[\{\pm0.981, \pm0.937, \pm0.872, \pm0.788, \pm0.686, \pm0.568, \pm0.438, \pm0.297, \pm0.150, 0\}.\] The E-optimal design is supported on these 21 points.
\end{example}

We remark that the use of high degree polynomials can also be circumvented using polynomial splines, which allow for the same large number of parameters without numerical difficulties; this will the subject of a forthcoming paper.

Finally, we present an example using rational functions.

\begin{example}[Measuring radiation parameters]\label{ex:numerical-radiation}
Consider the measurement of total radiation emitted from point sources, whose intensity obeys the inverse square law: $I_i(r)=\theta_i r^{-2}$ where $I_i$ is the intensity of the radiation emitted by source $i$ measured at distance $r$ from the source, for some unknown parameter $\theta_i$. The locations $x_i$ of the sources are known. The response variable in our model \eqref{eq:model} is the total radiation. To be estimated are the values $\theta_i$, affected by parameters of sources, shielding between the sources and detector, and several other factors. In this numerical example we consider a simple one-dimensional instance: the locations of the three sources are $x_1=-2$, $x_2 = 2$, $x_3 = 4$, and we are interested in the effective values of $\theta_i$ as measurable in the interval $[-1,1]$, where the variance of the measurement error and the parameters are assumed to be constant.

The distance of a detector at $t$ from the $i$th point source is $r_i = |t-x_i|$, so in our model \eqref{eq:model} we have $f_i = r_i^{-2} = (t-x_i)^{-2}, i=1,2,3$, and $\mathcal{I}=[-1,1]$. The solution of the semidefinite program, which took 0.2 seconds, yielded a three-point support for the E-optimal design: $\{-1, 0.231, 1\}$.
\end{example}

%

\section{Reconstructing the optimal design}\label{sec:reconstruction}

Once we obtained a non-zero polynomial $\pi$ from the optimal solution of \eqref{eq:main}, we can find the optimal design by solving a second semidefinite programming problem, using \mbox{Theorem \ref{thm:eperfagyi}}. But Theorem \ref{thm:main} is only useful if the polynomial $\pi$ in the optimal solution is not the zero polynomial. As the following example shows, in sufficiently degenerate cases it might be.

\begin{example}\label{ex:zeropoly}
Consider the E-optimal design problem when $m=2$, $\vf(t) = (1, t)^\T$, $\omega =1$, and $\mathcal{I} = [-1,1]$. Then the corresponding semidefinite programming problem simplifies to
\[ \min_{y,\vW}\; y\quad \textrm{s.t.}\; \vW\succcurlyeq 0,\; \tr(\vW) = 1,\; \pi = (y-W_{11}, -2W_{12}, -W_{22}) \in P^{[-1,1]}, \]
by essentially the same calculations as in \mbox{Example \ref{ex:numerical1}}.
It is not hard to see that the set of optimal solutions to this problem is
$\{ (y,\vW)\;|\; y = 1, W_{12} = 0, W_{11}+W_{22} = 1, 0\leq W_{11}\leq 1\}.$ Hence, we have infinitely many solutions, including $W_{11} = 1-W_{22} = 1$, which corresponds to $\pi(t) = 0$. Choosing any other optimal solution yields a polynomial whose roots are the expected $t=\pm 1$.

Alternatively, we can change $f$ to a different basis of degree one polynomials. This does not really change the model, however, if we choose, for example, $\vf(t) = (\alpha, t)^\T$ for any $\alpha > 1$, the above problem disappears: the resulting semidefinite programming problem has a unique optimal solution, and that solution corresponds to a nonzero polynomial $\pi$, with two real roots.
\end{example}

In the rest of the section we list a number of sufficient conditions that ensure that the optimal $\pi$ in \eqref{eq:final} is not the zero polynomial. The first one is perhaps the most obvious one.

\begin{lemma}
Let $f_1, \dots, f_m$ and $\omega$ in \eqref{eq:model} be chosen such that $1 \not\in \operatorname{span} \{\omega f_i f_j\;|\; 1\leq i \leq j \leq m \}$. Then no solution satisfying the constraints of \eqref{eq:final} has $\pi = 0$.
\end{lemma}

Special cases covered by this lemma include designs for incomplete polynomial models with no intercept, such as those considered in \cite{HCW-95} and \cite{CH-96}, and models involving rational functions, such as \mbox{Example \ref{ex:numerical-radiation}} above.

The last observation of \mbox{Example \ref{ex:zeropoly}} also generalizes to E-optimal designs for arbitrary polynomial systems.
\begin{lemma}
Consider the E-optimal design problem for a polynomial model with at least two parameters to be estimated. By choosing an appropriate basis $\{f_1, \dots, f_m\}$ in \eqref{eq:model} it can be guaranteed that no optimal solution of \eqref{eq:final} has $\pi = 0$.
\end{lemma}

\begin{proof}
Let $(\hat y,\hat \vW,\hat\pi)$ be an optimal solution to \eqref{eq:final}. Then $\hat \vW = \vY^\T \vY$ for some matrix $\vY$, and the polynomial $q\colon t\to \langle \hat \vW, \vM(\xi_t) \rangle$ can be written as $q(t) = \vz(t)^\T\vz(t)$ with $\vz(t) = \vY\vf(t)$. Consequently, $q$ can only be a constant (and $\hat\pi$ can only be the zero polynomial) if $\vz(t) = \vY \vf(t)$ is componentwise constant.

If $1 \not\in \operatorname{span} \{f_1,\dots,f_m\}$, then this is impossible, because $\vY=0$ is excluded by the constraints \eqref{eq:final-BC}, which simplifies to $\tr(\vW) = 1$ for E-optimal designs.

If $1 \in \operatorname{span} \{f_1,\dots,f_m\}$, then we can assume without loss of generality that $f_1 = 1$. Now $q$ can be a constant only if $\hat W_{11}=1$, and every other entry of $\hat \vW$ is zero, making $q(t) = 1$ and $\hat y=1$. Replacing $f_i$, $i \geq 2$ by $\lambda f_i$ with a sufficiently small positive $\lambda$ that satisfies $\lambda |f_i(t)| < 1$ for all $t\in\mathcal{I}$ ensures that this is not the optimal solution to \eqref{eq:final}.
\end{proof}

A similar argument applies to A-optimal designs for polynomial models. For brevity we omit the details. As above, one can argue that by scaling the non-constant basis functions, solutions to the semidefinite programming problem that yield constant zero $\pi$ cannot be optimal.
\begin{lemma}
Consider the A-optimal design problem for a polynomial model with at least two parameters to be estimated. By choosing an appropriate basis $\{f_1, \dots, f_m\}$ in \eqref{eq:model} it can be guaranteed that no optimal solution of \eqref{eq:final} has $\pi = 0$.
\end{lemma}

Finally, as a corollary to \mbox{Theorem \ref{thm:main}}, we also obtain an upper bound on the size of the support set of the minimally-supported optimal designs.

\begin{corollary}\label{cor:degree_bound}
Let $n_\omega$ and $d_\omega$ be the degree of the numerator and denominator of $\omega$, $n_i$ and $d_i$ be the degree of the numerator and denominator of $f_i$, and $d_{\den} = \lcm(d_\omega, d^2_1, \dots, d^2_p)$. Furthermore, suppose that $\mathcal{I}$ is the union of $k_1+k_2$ disjoint closed intervals, $k_1$ of which are singletons. (The remaining $k_2$ intervals have distinct endpoints.) Then for every admissible criterion $\Phi$ for which the optimal solution to \eqref{eq:final} does not have $\pi = 0$ there is a $\Phi$-optimal design supported on not more than
$\min (\tfrac{1}{2}(k_1 + 2k_2 + \deg\pi), \deg\pi)$
points, where $\deg \pi = d_{\den}+(n_\omega - d_\omega + 2 \max_i (n_i - d_i))_+$.
\end{corollary}
\begin{proof}
We need to count the number of distinct zeros of the polynomial $\pi$ in \eqref{eq:final}. On one hand, $\pi$ cannot have more than $\deg \pi$ roots. On the other hand, since $\pi$ is nonnegative over $\mathcal{I}$, each of its zeros must be either an endpoint of an interval constituting $\mathcal{I}$ or a root of multiplicity at least two. Hence the number of distinct zeros of $\pi$ is at most $\tfrac{1}{2}(k_1 + 2k_2 + \deg\pi)$. Finally, the expression for $\deg\pi$ comes directly from \eqref{eq:def_pi}.
\end{proof}

\section{Parameter subsystems, estimability}\label{sec:subsystems}

Often the experimenter is not interested in the entire parameter vector $\theta$, but rather in a subset of them, or more generally in $s\leq m$ specific linear combinations of the parameters: $\vk_j^\T \theta$, $j=1, \dots, s$. Let $\vK$ be the matrix whose columns are $\vk_1, \dots, \vk_s$; so far we have assumed $s=m$ and $\vK=\vI$. An application of this more general setting is polynomial regression, when the experimenter needs to test whether the highest degree terms in the model are indeed non-zero.

It can assumed without loss of generality that $\vK$ has full (column) rank, and to make the problem meaningful, it must be assumed that the parameters $\vK^\T\theta$ are \emph{estimable}, that is, \begin{equation}\label{eq:estimability}\operatorname{range}(\vK) \subseteq \operatorname{range}(\vM),\end{equation}
see for example \cite[Chapter 3]{Pukelsheim-93}. In this setting the matrix $\vM$ is replaced by the information matrix $(\vK^\T \vM^\psinv \vK)^{-1}$, where $\vM^\psinv$ denotes the Moore--Penrose pseudo-inverse of $\vM$. In particular, the optimal design is a probability measure $\hat\xi$ that maximizes the matrix $(\vK^\T \vM^\psinv(\xi) \vK)^{-1}$, or the function \mbox{$\xi \to \Phi\big((\vK^\T \vM^\psinv(\xi) \vK)^{-1}\big)$} for some criterion function $\Phi$ compatible with the L{\"o}wner partial order.

The optimization models for this setting can be developed analogously to the model of the previous section. Since $\Phi$ is assumed to be compatible with the L{\"o}wner partial order, $\max_{\vM\in\Mm} \Phi\big((\vK^\T \vM^\psinv(\xi) \vK)^{-1}\big)$ is equivalent to
\begin{equation}\label{eq:model_with_K_1}
\max \{ \Phi(\vY)\;|\; \vM\in\Mm, (\vK^\T \vM^\psinv \vK)^{-1} \succcurlyeq \vY \succcurlyeq 0\}.
\end{equation}
Note that the optimum does not change if we require $\vY$ to be positive definite, in which case the last two inequalities are equivalent to $\vY^{-1} \succcurlyeq \vK^\T \vM^\psinv \vK$. We shall use now a Schur complement characterization of semidefinite matrices, which is a generalization of the result used in \mbox{Example \ref{ex:A-optimality}}.

\begin{proposition}[\protect{\cite[\mbox{Theorem 1.20}]{Zhang-05}}]\label{prop:SchurPSD_general_form}
The symmetric block matrix $\left(\begin{smallmatrix} \vM & \vK\\ \vK^\T & \vZ \end{smallmatrix}\right)$ is positive semidefinite if and only if
$\vM \succcurlyeq 0$, $\vZ \succcurlyeq \vK^\T \vM^\psinv \vK$, and $\operatorname{range}(\vK) \subseteq \operatorname{range}(\vM)$.
\end{proposition}


By this proposition, \eqref{eq:model_with_K_1} is equivalent to
\begin{equation*}
\max \{ \Phi(\vY)\;|\; \vM\in\Mm,\; \left(\begin{smallmatrix} \vM & \vK\\ \vK^\T & \vY^{-1} \end{smallmatrix}\right) \succcurlyeq 0\}.
\end{equation*}
Using Schur complements again, the inversion from the last inequality can be eliminated, and we obtain the following equivalent optimization problem:
\begin{equation}\label{eq:model_with_K_3}
\max \{ \Phi(\vY)\;|\; \vM\in\Mm,\; \vM \succcurlyeq \vK \vY \vK^\T,\; \vY \succcurlyeq 0\}.
\end{equation}

Finally, we can simplify this problem essentially identically to how we obtained \eqref{eq:final} from \eqref{eq:orig}. Doing so we obtain the following.

\begin{theorem}\label{thm:main-with-K}
Consider the linear model \eqref{eq:model} and a matrix $\vK \in \real^{m\times s}$ satisfying $\rank(\vK) = s$ and the estimability condition \eqref{eq:estimability}. Then for every semidefinite representable criterion function $\Phi$ a polynomial $\pi$ whose real zeros contain the support of a $\Phi$-optimal design for the parameter vector $\vK^\T\theta$ is an optimal solution of the following semidefinite program:
\begin{align*}
\mathop{\mathrm{minimize}}_{\substack{y\in\real, \vV \in \Smp, \pi\in\real^d,\\ \vW_1, \dots, \vW_p \in \Skp}}\;\; & y \\
\mathrm{subject \; to}\;\;\;\; & \vK^\T \vV\vK \succcurlyeq \sum_{i=1}^p A_i^*(\vW_i),\\
                      & \sum_{i=1}^p \langle \vW_i, \vB_i \rangle = -1,\quad \sum_{i=1}^p C_i^*(\vW_i) = 0,\\
                      & \pi = \Pi(y, \vV, \vW_1, \dots, \vW_p) \in P^\mathcal{I},
\end{align*}
where $A_i, \vB_i, C_i$ and $\vD_i$ come from \mbox{Definition \ref{def:SD-representable-func}}, and $d$ is the degree of the polynomial
\begin{equation*}
t \to \lcm(\den(\omega), \den(f_1^2), \dots, \den(f_p^2)) \bigg(y - \langle \vV, \vM(\xi_t) \rangle - \sum_{i=1}^p \langle \vW_i, \vD_i\rangle\bigg),
\end{equation*}
whose coefficient vector is denoted by $\Pi(y, \vV, \vW_1, \dots, \vW_p)$.
\end{theorem}

We omit the rest of the proof as it is essentially identical to that of \mbox{Theorem \ref{thm:main}}, given in the Appendix. The main difference is the appearance of the variable $\vV$, which is the dual variable of the constraint $\vM \succcurlyeq \vK\vY\vK^\T$.


\section{Locally optimal designs for nonlinear models}\label{sec:nonlinearmodels}
In this section we show how to apply \mbox{Theorems \ref{thm:main} and \ref{thm:main-with-K}} to find \emph{locally optimal designs} (with respect to various optimality criteria) for nonlinear rational models (see the definition and its motivation below). We consider the general nonlinear model \begin{equation}\label{eq:model-nonlinear} y(t) = f(t;\theta) + N(0,\sigma(t)), \qquad t \in \mathcal{I}, \end{equation}
where $f$ is a rational function of $(t;\theta)$, $\theta$ is an $m$-vector of unknown parameters. The designs space $\mathcal{I}$ is the union of finitely many closed intervals, as before.

Nonlinear regression models are widely used and researched, but finding optimal designs for nonlinear regression is particularly challenging -- so much so that even numerical solutions to simple two- and three-variable models are highly non-trivial to obtain. (See for example \cite{DBPP-08} for recent results on a number of models used in dose-finding studies, and \cite{LM-09} for pharmacokinetic models.) Nonlinear rational models (where the response variable is a rational function of the explanatory variable and the unknown parameters) and models involving exponential functions and logarithms are particularly well studied. Imhof and Studden \cite{IS-01} considered E-optimal designs for different classes of rational models. More recently, Dette \emph{et al.} \cite{DMP-04} investigated E-optimal designs for a more general family of functions (not only rational functions), under the assumption that some partial derivatives of the model function form a \emph{weak Chebyshev system} \cite{KS-66}. Note that this class of problems is not comparable to the rational models we are considering: the partial derivatives of many non-rational functions satisfy this criterion, but many rational models, for instance, the $E_{max}$ model from \mbox{Example \ref{ex:emax}} below, are outside that class.

Perhaps the most fundamental complication in designing non-sequential experiments for nonlinear models is in the formulation of the problem as a meaningful optimization problem. For a nonlinear regression model \eqref{eq:model-nonlinear} the Fisher information matrix corresponding to the design $\xi$ is
\begin{equation}\label{eq:Fisher-nonlinear}
M(\xi, \theta) = \int_{\mathcal{I}} (\partial f(t,\theta)/\partial \theta) (\partial f(t,\theta)/\partial \theta)^\T \omega(t) d \xi(t).
\end{equation}
It is immediate that (unlike in the linear case) the Fisher information matrix for nonlinear models depends on the parameters whose estimation is the purpose of the experiments we are to design. Hence defining the optimal designs as the optimizers of the $M(\xi, \theta)$ is meaningless. Nevertheless, if the experimenter can guess reasonable values of the parameters, it can be useful to design the experiment that would be optimal if the guessed parameters were correct. Some more advanced design methods, such as sequential designs \cite{FS-80} also build on the same concept, often called \emph{locally optimal designs}. (The same ideas can also be used for the estimation of nonlinear functions of the parameters of a linear model.) Before considering the general case, let us look at a simple example that we shall readily generalize below.

\begin{example}\label{ex:emax}Consider the three-parameter $E_{max}$ model
\begin{equation}\label{eq:Emax-model}y(t) = \theta_0 + \frac{\theta_1 t}{t + \theta_2} + N(0,1),\end{equation}
from the dose-finding study \cite{DBPP-08}. With the notation of \eqref{eq:Fisher-nonlinear},
\[\frac{\partial f(t,\theta)}{\partial \theta} = \big(1, t(t+\theta_2)^{-1}, -\theta_1t(t+\theta_2)^{-2} \big)^\T,\]
so for every fixed value $(\theta_0^*, \theta_1^*, \theta_2^*)$ of $\theta$ the integrand in the Fisher information matrix \eqref{eq:Fisher-nonlinear} can be written as
\begin{equation}\label{eq:integrand_1} \begin{pmatrix}
1 & \frac{t}{t+\theta_2^*} & -\frac{\theta_1^* t}{(t+\theta_2^*)^2}\\
\frac{t}{t+\theta_2^*} & \frac{t^2}{(t+\theta_2^*)^2} & -\frac{\theta_1^* t^2}{(t+\theta_2^*)^3}\\
-\frac{\theta_1^* t}{(t+\theta_2^*)^2} & -\frac{\theta_1^* t^2}{(t+\theta_2^*)^3} & \frac{(\theta_1^* t)^2}{(t+\theta_2^*)^4}
\end{pmatrix},\end{equation}
which is the same information matrix as the information matrix of the parameter vector $(\alpha_0, \alpha_1, \alpha_2)$ for the linear model
\begin{equation}\label{eq:Emax-linear-ugly}y(t) = \alpha_0 + \alpha_1 \frac{t}{t + \theta_2^*} +  \alpha_2\frac{\theta_1^* t}{(t + \theta_2^*)^2} + N(0,1). \end{equation}
Hence, finding locally optimal designs for the $E_{max}$ model \eqref{eq:Emax-model} is equivalent to finding optimal designs for the linear model \eqref{eq:Emax-linear-ugly}, which is a linear model with rational functions, hence \mbox{Theorem \ref{thm:main}} is applicable.

A further simplification is possible: we can find an equivalent polynomial model, and use \mbox{Theorem \ref{thm:main-with-K}} to find optimal designs. It is easy to verify that the matrix \eqref{eq:integrand_1} can also be written as
\[K^\T \begin{pmatrix}1 & \chi & \chi^2\\ \chi & \chi^2 & \chi^3\\ \chi^2&\chi^3&\chi^4\end{pmatrix} K\]
with $\chi = (t+\theta_2^*)^{-1}$ and $K = \left(\begin{smallmatrix}1 & 1 & 0 \\ 0 & -\theta_2^* & -\theta_1^* \\ 0 & 0 & \theta_1^* \theta_2^*\end{smallmatrix}\right)$. Hence, for every fixed $\theta^*$ the Fisher information matrix of the design $\xi$ for model \eqref{eq:Emax-model} is identical to the Fisher information matrix of the design that puts $\xi(t_i)$ mass to the point $\chi_i = (t_i+\theta_2^*)^{-1}$ for the three-parameter linear model
\begin{equation} \label{eq:Emax-linear} y(\chi) = \alpha_0 + \alpha_1\chi +  \alpha_2\chi^2 + N(0,1)\end{equation}
and the parameter vector $K^\T\alpha = (\alpha_0,\alpha_0-\alpha_1 \theta_2^*,\alpha_2 \theta_1^* \theta_2^*-\alpha_1 \theta_1^*)^\T$. Now the problem is reduced to polynomial regression, and \mbox{Theorem \ref{thm:main-with-K}} is applicable.
\end{example}


Generally, for a nonlinear regression model \eqref{eq:model-nonlinear}
with $m$ parameters, the problem of finding a locally optimal design for a given parameter vector $\theta^*$ is equivalent to finding the optimal design for the associated linear model of the form \eqref{eq:model} with $f_i = (\partial f)/(\partial \theta_i)\big|_{\theta = \theta^*}$, $i=1, \dots, m$. If $f$ is a rational function of $(t,\theta)$, then so are its partial derivatives. Hence the equivalent linear model (for every fixed value of $\theta$) is always one with rational functions $f_i$.

The same observation was used in \cite{DMP-04} to derive E-optimal designs for the class of nonlinear regression models where the partial derivatives form a weak Chebyshev system. Now this assumption can be dropped, and other optimality criteria can also be considered.

\section{Optimal designs in other functional spaces}\label{sec:othersystems}

Most of \mbox{Section \ref{sec:design}} applies to every $f_i$ and $\omega$, not only to rational functions; for example, \eqref{eq:dual} is not specific to polynomials or rational functions. As long as the set of constraints \eqref{eq:rat_ineq} can be expressed by finitely many semidefinite constraints (or in any other computationally tractable manner), the same approach works. Examples include the following (we consider only the homoscedastic case for simplicity):
\begin{enumerate}
\item $f_i(t) = \cos(it)$ for every $i\in\mathbb{N}$ and $t$;
\item $f_{2i}(t) = \cos(it)$, $f_{2i+1}(t) = \sin(it)$ for every $i\in\mathbb{N}$, and $t$;
\item $f_i(t) = \exp(it)$ for every $i\in\mathbb{N}$ and $t$.
\end{enumerate}
These three examples, however, do not truly generalize the approach of Section \ref{sec:design}, since they can also be reduced to the polynomial case by an appropriate change of variables. (We omit the details.)

Our estimate on the number of support points is also valid for some functional spaces other than polynomials. The only property of polynomials that we used were that their degree bounds the number of their roots (counted with multiplicity: roots in the interior of the domain have multiplicity two). Hence, bounds similar to the one in \mbox{Corollary \ref{cor:degree_bound}} can be obtained for models where the functions $\{\omega f_i f_j|i,j\}$ form a Chebyshev system.

%

\section{Discussion}

Computing optimal designs for linear models involving rational functions is easy when the design space is finite, hence the key difficulty in obtaining optimal designs for infinite design spaces, such as intervals or unions of intervals, is that the finite support of an optimal design has to be determined. Symbolic or closed form solutions are unavailable for most models, and their scope is often limited by assumptions that are neither technical, nor have any statistical interpretation. In this paper, we have presented a method that does not rely on such assumptions. It is an effective method to determine a polynomial whose zeros contain the support of the optimal designs. The method is applicable to every linear regression problem involving only rational functions; it treats D-, A-, E-, and general $\Phi_p$ optimal designs in a unified manner, and generalizes to the heteroscedastic case if the variance of the noise is a positive rational function. The design space can be an interval or the union of finitely many intervals.

This level of generality is far greater than what appears to be possible by closed-form approaches. It is achieved at the price of providing numerical, rather than symbolic, solutions: the method generates the (numerical) coefficients of the sought polynomial. The main step of the method is the solution of a semidefinite programming problem, which can be done (to high precision) with readily available software in trivial running time. Unlike other iterative methods previously proposed in the literature, including all of those based on coordinate descent, semidefinite programming algorithms have a theoretically guaranteed low running time, and are guaranteed to find the globally optimal design, rather than a local optimum. This is of considerable practical importance, with the potential for instance to impact design software development.

Further study of the optimality conditions of the main optimization model might also yield new theoretical insights.

Through a number of examples we have demonstrated the flexibility of the proposed method, and we also found that the algorithm is robust enough to handle ill-conditioned problems involving high-degree polynomials, and yields solutions in a fraction of a second for problems of practical size.

A corollary of our main theorem is a bound on the size of the support set, and an analogous optimization model for the estimation of parameter subsystems.

Most results of this paper readily generalize to linear models involving certain exponential families rather than rational functions; these include Fourier regression, where the model is a trigonometric polynomial with unknown coefficients.
The method may also be used to find locally optimal designs for nonlinear models. In this area almost no symbolic solutions are available, but model-specific numerical methods are abound. Details are available from the author, and may be subject of a future paper.

A few important questions remain open. The first one is how to extend the results of \mbox{Section \ref{sec:reconstruction}}. Since the optimal solution to the problem \eqref{eq:final} is sensitive to both the representation of the optimality criterion $\Phi$ and also to the basis $\{f_i\}$ of the space of regression functions (meaning that equivalent representations of $\Phi$ and basis transformations lead to different optimal solutions), one may readily conjecture that for every model \eqref{eq:model} and for every admissible optimality criterion one can find an equivalent model (that is, a basis $\{f_i\}$ of the same functional space) and a semidefinite representation \eqref{eq:SDP-rep_def} for $\Phi$ such that the optimal $\pi$ in every solution of \eqref{eq:final} is nonzero.

Another subject of future research may be the generalization of our results to larger classes of functions. Chebyshev systems are natural candidates to look at, but more importantly, the ideas of the paper would generalize word by word to every family $(f_1, \dots, f_m)$ and weight function $\omega$ for which functions in the space $\operatorname{span}\{\omega f_i f_j|i,j\}$ are easy to maximize. Hence, identifying such spaces of functions would be particularly important.

Finally, the ability to design experiments in a discontinuous design space is extremely relevant in practice, especially in the multivariate case (e.g., when measurements cannot be taken at inaccessible locations, or are practically impossible very close to signal sources). Existing models with closed-form solutions are not applicable, and most of the current numerical methods cannot address this problem even in the univariate case, aside from sporadic results involving two disjoint intervals for a few concrete models.

The applicability of the proposed method in the multivariate setting also requires further study.

\singlespacing
\bibliographystyle{amsplain}
\bibliography{design}

\appendix

\section{The semidefinite representability of polynomials over intervals}
For a $\Delta \subseteq \real$ let $\Po_n^\Delta$ denote the set of degree $n$ polynomials nonnegative over $\Delta$. The following representation of nonnegative polynomials is well-known:
\begin{proposition}[\cite{Luk-18}]\label{prop:WSOS-poly-char} For every polynomial $p$ of degree $n$,
$p \in \Po_{n}^{[a,b]}$ if and only if
    \[p(t) = \begin{cases} r^2(t) + (t-a)(b-t) q^2(t) & (\text{if } n=2k) \\ (t-a) r^2(t) + (b-t) s^2(t) & (\text{if } n=2k+1) \end{cases}\]  for some polynomials $r$ and $s$ of degree $k$ and $q$ of degree $k-1$.
\end{proposition}

On the other hand, functions expressible as sums of squares of functions from a given finite dimensional functional space (such as polynomials of a fixed degree) are semidefinite representable; see \cite{nes-00} for a constructive proof of this claim. Applying this construction to part (2) of Proposition \ref{prop:WSOS-poly-char} yields the following.
\begin{proposition}[\cite{nes-00}]\label{prop:WSOS-poly-char-sdp}
Suppose $p$ is a polynomial of degree $n = 2m+1$, $p(t) = \sum_{k=0}^{n} p_k t^k$, and let $a < b$ are real numbers. Then $p \in \Po_{n}^{[a,b]}$ if and only if there exist positive semidefinite matrices
    $\vX = (x_{ij})_{i,j=0}^m$ and $\vY = (y_{ij})_{i,j=0}^m$
    satisfying
    \begin{equation}\label{eq:OddDegreePolyChar}
        p_k = \sum_{i+j=k} (-a x_{ij} + b y_{ij}) + \sum_{i+j=k-1}      (x_{ij}-y_{ij})
    \end{equation}
    for all $k = 0, \dots, 2m+1$.

    Similarly, if $p$ is a polynomial of degree $n = 2m$, then $p \in \Po_{n}^{[a,b]}$ if and only if there exist positive semidefinite matrices $\vX = (x_{ij})_{i,j=0}^m$ and $\vY = (y_{ij})_{i,j=0}^{m-1}$
    satisfying
    \begin{equation}\label{eq:EvenDegreePolyChar}
        p_k = \sum_{i+j=k} (x_{ij} - ab y_{ij}) + \sum_{i+j=k-1} (a+b)y_{ij} - \sum_{i+j=k-2} y_{ij}
    \end{equation}
    for all $k = 0, \dots, 2m$.
\end{proposition}

This is rather involved (and the details are only important for the purposes of actual computations), but close inspection reveals that this proposition characterizes $\Po_{n}^{[a,b]}$ as a linear image of the Cartesian product of two semidefinite cones, thus, it proves the semidefinite representability of $\Po_n^{[a,b]}$ in the sense of Definition \ref{def:SD-representable-set}.

Since the intersection of semidefinite representable sets are also semidefinite representable, it follows that $\Po_n^{\mathcal{I}}$ is semidefinite representable for every union of finitely many closed intervals $\mathcal{I}$.

\section{Proof of Theorem \ref{thm:main}}

Consider the problem of finding $\max \{ \Phi(\vM(\xi))\,|\,\xi\in\Xi(\mathcal{I}) \}$, where $\Xi(\mathcal{I})$ is the set of probability measures on $\mathcal{I}$ with finite support, and $\vM$ is the Fisher information matrix defined by \eqref{eq:Fisher}. Considering the Fisher information as the variable, this can be expressed as a finite dimensional optimization problem:
\begin{equation}\label{eq:orig} \max \{ \Phi(\vM)\,|\, \vM \in \Mm \},\quad \text{where } \Mm = \{ \vM(\xi)\,|\,\xi\in\Xi(\mathcal{I}) \}.
\end{equation}

Let $\xi_t$ be the probability measure that assigns all of its mass to \mbox{$t\in\mathcal{I}$}. Because $\mathcal{I}$ is assumed to be compact and the mapping $t \to \vM(\xi_t)$ is continuous, $\{ \vM(\xi_t)\,|\,t\in \mathcal{I} \}$ is compact. Hence, $\Mm = \conv\{ M(\xi_t)\,|\,t\in \mathcal{I} \}$ is a convex compact set, and the optimization problem \eqref{eq:orig} is well-defined: The maximum is finite, and is attained (for every continuous function $\Phi$).

Now let us assume that $\Phi$ is semidefinite representable. Then using the notations of \mbox{Definition \ref{def:SD-representable-func}}, problem \eqref{eq:orig} may be written as follows.
\begin{equation}\label{eq:main}
\begin{split}
\mathop{\text{maximize}}_{z\in\real,\, \vu\in\real^l,\, \vM \in \Mm} & z \\
\text{subject to} \;\; &  A_i(\vM) + \vB_iz + C_i(\vu) + \vD_i \succcurlyeq 0, \quad i = 1, \dots, p,
\end{split}
\end{equation}
where $A_i, \vB_i, C_i$, and $\vD_i$ are the functions and matrices as in \mbox{Definition \ref{def:SD-representable-func}}.

Because $\Skp$ is a closed convex cone, \eqref{eq:main} is equivalent to the following Lagrangian relaxation (in which the dual variable $\vW_i$ is the Lagrange multiplier associated with the $i$th constraint):
\begin{equation}\label{eq:Lagrangian-P}
\max_{z,\vu,\vM\in\Mm}\; \inf_{\substack{\vW_i\succcurlyeq0\\ (i=1,\dots,p)}}\; z+\sum_{i=1}^p \langle \vW_i, A_i(\vM)+\vB_iz+C_i(\vu)+\vD_i\rangle
\end{equation}

Suppose that $\Phi$ is admissible with respect to $\Mm$. Then the optimization problem \eqref{eq:main} has a Slater point, consequently its optimum is equal to optimum of its dual problem \cite[Chapter 4]{Rusz-05}, obtained by replacing the ``$\max \inf$'' by ``$\min \sup$'' in the Lagrangian \eqref{eq:Lagrangian-P}. This dual problem then can be simplified as follows ($C_i^*$ denotes the dual operator of $C_i$):

\begin{align}
& \min_{\vW_1, \dots, \vW_p\succcurlyeq0}\; \sup_{z,\vu,\vM\in\Mm}\; z+\sum_{i=1}^p \langle \vW_i,A_i(\vM)+\vB_iz+C_i(\vu)+\vD_i\rangle \notag \\
& = \min_{\vW_1, \dots, \vW_p\succcurlyeq0}\; \sup_{z,\vu,\vM\in\Mm}\; z \left(1+\sum_{i=1}^p \langle \vW_i, \vB_i\rangle\right) + \notag \\ & \qquad\quad + \sum_{i=1}^p \langle C_i^*(\vW_i), \vu\rangle + \sum_{i=1}^p \langle \vW_i, A_i(\vM) + \vD_i\rangle \notag \\
& = \min_{\substack{\vW_1, \dots, \vW_p\succcurlyeq0\\ \sum_i \langle \vW_i, \vB_i \rangle = -1 \\ \sum_i C_i^*(\vW_i) = 0}}\; \sup_{\vM\in\Mm}\; \sum_{i=1}^p \langle \vW_i, A_i(\vM) + \vD_i\rangle \label{eq:dual-ugly} \\
& = \min_{\substack{\vW_1, \dots, \vW_p\succcurlyeq0\\ \sum_i \langle \vW_i, \vB_i \rangle = -1 \\ \sum_i C_i^*(\vW_i) = 0}}\; \max_{\vM\in\Mm}\; \sum_{i=1}^p \langle \vW_i, A_i(\vM) + \vD_i\rangle.\notag
\end{align}
(The last equation simply means that the supremum is attained.)

Finally, with the help of a dummy variable $y$ the optimization problem in the last line can be conveniently written as:
\begin{equation}\label{eq:dual}\begin{aligned}
\mathop{\text{minimize  }}_{\substack{y\in\real,\\ \vW_1, \dots, \vW_p \in \Sk}} &\;\; y \\
    \text{subject to  } & \;\;\vW_i \succcurlyeq 0 &i=1,\dots,p,\\
                        & \;\;\sum_{i=1}^p \langle \vW_i, \vB_i \rangle = -1, \quad \sum_{i=1}^p C_i^*(\vW_i) = 0,\\
                        & \;\;y \geq \sum_{i=1}^p \langle \vW_i, A_i(\vM)+\vD_i \rangle &\forall\,\vM\in\Mm.
\end{aligned}
\end{equation}

Aside from the last set of constraints, which is an uncountably infinite collection of linear inequalities, every constraint is either a linear equality or a linear matrix inequality on the variables $\vW_i$. Using that $\Mm = \conv \{ \vM(\xi_t)\,|\,t\in \mathcal{I} \}$, the last set of constraints can also be simplified to
\begin{equation}\label{eq:rat_ineq} y - \sum_{i=1}^p \langle \vW_i, A_i(\vM(\xi_t))+\vD_i \rangle \geq 0 \qquad \forall\,t \in \mathcal{I}. \end{equation}
Since $\vM(\xi_t)$ is a matrix whose entries are rational functions of $t$, this inequality  expresses the nonnegativity of a rational function (over $\mathcal{I}$) that lives in the space
\begin{equation*}
V = \operatorname{span}\left(\{ \omega f_i f_j\,|\,i,j=1,\dots,m \} \cup \{1\} \right),
\end{equation*}
with variable coefficients. Multiplying both sides with the least common denominator of the functions $\omega f_i f_j$ (which is positive on $\mathcal{I}$) turns \eqref{eq:rat_ineq} to the equivalent inequality \eqref{eq:POP-constr} with $\Pi$ defined in \eqref{eq:def_pi}, giving us \eqref{eq:final}.

Finally, suppose $(\hat y, \hat\pi, \hat \vW_1, \dots, \hat \vW_p)$ is an optimal solution to \eqref{eq:final}. Then, since \eqref{eq:final}, \eqref{eq:dual-ugly}, and \eqref{eq:dual} are equivalent, $(\hat y, \hat \vW_1, \dots, \hat \vW_p)$ is also an optimal solution to \eqref{eq:dual}, and because the optimum in \eqref{eq:dual-ugly} is attained, there exists an $\hat \vM\in \Mm$ that satisfies the last constraint of \eqref{eq:dual} with inequality. The way we obtained \eqref{eq:dual} from \eqref{eq:orig} ensures that this $\hat \vM$ is also an optimal solution to our original problem \eqref{eq:orig}. Suppose $\hat \vM = \hat \vM(\hat \xi)$ for some measure $\hat \xi \in \Xi(\mathcal{I})$ that is concentrated on $\{t_1, \dots, t_k\}\subseteq \mathcal{I}$ and assigns weight $\lambda_i$ to $t_i$, $i=1, \dots, k$. Then with the optimal $\hat y$ and $\hat \vW_1, \dots \hat \vW_p$ each of these $t_i$ must satisfy \eqref{eq:rat_ineq} with equality. Consequently each $t_i$ is a root of $\hat \pi$. \qed

\end{document}